\documentclass[11pt]{article}
\usepackage{mathrsfs} 
\usepackage[margin=1.25in]{geometry}
\usepackage{ifthen}
\usepackage{graphicx}
\usepackage{epsfig}
\usepackage{url}
\usepackage{dsfont, amssymb, appendix}
\usepackage{amsfonts,dsfont,amssymb,amsthm,stmaryrd,bbm}
\usepackage[cmex10]{amsmath} % Use the [cmex10] option 
\usepackage{hyperref,mathtools}
\hypersetup{colorlinks=true,pdftitle=""}

\newtheorem{conj}{Conjecture}[section]
\newtheorem{thm}{Theorem}[section]

\newtheorem{lem}[conj]{Lemma}
\newtheorem{prop}[conj]{Proposition}

\newtheorem{defn}[conj]{Definition}
\newtheorem{cor}[conj]{Corollary}

\newcommand\independent{\protect\mathpalette{\protect\independent}{\perp}} 
\def\independent#1#2{\mathrel{\rlap{$#1#2$}\mkern2mu{#1#2}}}

\DeclareMathOperator{\Var}{Var}

\def\Var{{\rm Var}}

\def\phi{\varphi}

\def\bee{\begin{eqnarray*}}
\def\ene{\end{eqnarray*}}

\title{Strongly Convex Divergences}
\author{James Melbourne}

\begin{document}

\maketitle

\begin{abstract}  
    We consider a sub-class of the $f$-divergences satisfying a stronger convexity property, which we refer to as strongly convex, or $\kappa$-convex divergences. 
    We derive new and old relationships, based on convexity arguments, between popular $f$-divergences.
\end{abstract}

\section{Introduction}

The concept of an $f$-divergence, introduced independently by Ali-Sivley \cite{AS66} and Csizis\'ar \cite{Csi63}, unifies several important information measures between probability distributions, 
as integrals of a convex function $f$, composed with the Radon-Nikodym of the two probability distributions.
For a convex function $f:(0,\infty) \to \mathbb{R}$ such that $f(1) = 0$, 
and measures $P$ and $Q$ such that $P \ll Q$ the $f$-divergence from $P$ to $Q$ is given by
$ 
   D_f(P||Q) \coloneqq \int f\left(\frac{dP}{dQ}\right) dQ.
$
The canoncial example of an $f$-divergence, realized by taking $f(x) = x \log x$, is the relative entropy (often called the KL-divergence), and $f$-divergences inherit 
many properties enjoyed by this special case; non-negativity, joint convexity of arguments, and a data processing inequality.
Other important examples include the total variation, the $\chi^2$-divergence, and the squared Hellinger distance.  
The reader is directed to Chapter 6 and 7 of \cite{polyanskiy2014lecture} for more background.

We will be interested in how stronger convexity properties of $f$ give improvements of classical $f$-divergence inequalities.  
This is in part inspired by the work of Sason \cite{sason2019data}, who demonstrated that divergences that are (as we define later) ``$\kappa$-convex'' 
satisfy ``stronger than $\chi^2$'', data-procesing inequalities.  

Aside from the total variation, most divergences of interest have stronger than affine convexity, at least when $f$ is restricted to a sub-interval of the real line.
This observation is especially relevant to the 
situtation in which one wishes to study $D_f(P||Q)$ in the existence of a bounded Radon-Nikodym derivative $\frac{dP}{dQ} \in (a,b) \subsetneq (0,\infty)$.   One naturally obtains such bounds 
for skew divergences.  That is divergences of the form $(P,Q) \mapsto D_f( (1-t)P + t Q|| (1-s)P + s Q)$ for $t,s \in [0,1]$, as in this case, $\frac{(1-t) P + t Q}{(1-s)P + s Q} \leq \max \left\{ \frac{1-t}{1-s}, \frac{t}{s} \right\}$.  Important examples of skew-divergences include the skew divergence 
\cite{lee1999measures} based on the relative entropy and the Vincze-Le Cam divergence \cite{vincze1981concept, LeC86:book}, called the triangular discrimination in \cite{topsoe2000some}
 and its generalization due to Gy\"orfi and Vajda \cite{gyorfi2001class} based on the
$\chi^2$-divergence.  The Jensen-Shannon divergence \cite{lin1991divergence} and its recent generalization \cite{nielsen2020generalization} give examples of $f$-divergences
realized as linear combinations of skewed divergences.  

Let us outline the paper.  In Section \ref{sec: Strongly convex divergences} we derive elementary results of $\kappa$-convex divergences and give a table of examples of $\kappa$-convex divergences.
 We demonstrate that 
$\kappa$-convex divergences can be lower bounded by the $\chi^2$-divergence, and that the joint convexity of the map $(P,Q) \mapsto D_f(P||Q)$ can be sharpened under $\kappa$-convexity
conditions on $f$.  As a consequence we obtain bounds between the mean square total variation distance of a set of distributions from its barycenter, and the average $f$-divergence from the set 
to the barycenter.

In Section \ref{sec: Skew divergences} we investigate general skewing of $f$-divergences.  In particular we introduce the skew-symmetrization of an $f$-divergence, which 
recovers the Jensen-Shannon divergence and the Vincze-Le Cam divergences as special cases.  We also show that a scaling of the Vincze-Le Cam divergence is minimal among skew-symmetrizations
of $\kappa$-convex divergences on
$(0,2)$.  We then consider linear combinations of skew divergences, and show that a generalized Vincze-Le Cam divergence (based on skewing the $\chi^2$-divergence) can be upper bounded by the generalized Jensen-Shannon divergence
 introduced
recently by Neilsen \cite{nielsen2020generalization} (based on skewing the relative entropy), reversing the obvious bound that can be obtained from the classical bound $D(P||Q) \leq \chi^2(P||Q)$.  
We also derive upper and lower total variation bounds for Neilsen's generalized Jensen-Shannon divergence.

In Section \ref{sec: TV bounds and Bayes risk} we consider a family of densities $\{p_i\}$ weighted by $\lambda_i$, and a density $q$.  
We use the Bayes estimator\footnote{This is the Bayes estimator for the loss function $\ell(i,j) = 1- \delta_i(j)$} $T(x) = \arg \max_i \lambda_i p_i(x)$ to derive 
a convex decomposition of the barycenter $p = \sum_i \lambda_i p_i$ and of $q$, each into 
two auxiliary densities.  We use this decomposition to sharpen, for $\kappa$-convex divergences, an elegant theorem of Guntuboyina \cite{guntuboyina2011lower} 
that generalizes Fano and Pinsker's inequality to 
$f$-divergences.  We then demonstrate explicitly, using an argument of Topsoe, how our sharpening of Guntuboyina's inequality gives a new sharpening of Pinsker's inequality
in terms of the convex decomposition induced by the Bayes estimator.

\subsection*{Notation}
We consider Borel probability measures $P$ and $Q$ on a Polish space $\mathcal{X}$.
For a convex function $f$ such that $f(1) = 0$,
define the $f$-divergence from $P$ to $Q$, via densities $p$ for $P$ and $q$ for $Q$ with respect to a common reference measure $\mu$ as
\begin{align}   
    D_f(p||q) &= \int_{\mathcal{X}} f\left( \frac{p}{q} \right) q d\mu
        \\
        &=
            \int_{\{pq > 0\}} q f\left( \frac p q \right) d\mu + f(0) Q( \{ p = 0\})  + f^*(0) P( \{ q = 0 \}).
\end{align}

We note that this representation is independent of $\mu$, and such a reference measure always exists, take $\mu = P+Q$ for example.

%For $\kappa \in (-\infty, \infty)$, we call a function $f$, $\kappa$-convex on an interval $I \subseteq \mathbb{R}$ when the function $x \mapsto f(x) - \kappa x^2/2$ is convex for $x \in I$.  Notice that $0$-convexity corresponds to ordinary convexity.  
%For numbers $t$ and $s$, define $t \wedge s = \min \{t, s\}$ and $t \vee s = \max \{ t, s\}$, then 
For $t,s \in [0,1]$, define
\begin{align}
    D_f(t||s) \coloneqq s f\left( \frac{t}{ s }\right) + (1-s) f \left( \frac{1-t}{1-s} \right)
\end{align}
with the conventions, $f(0) = \lim_{t \to 0^+} f(t)$, $0 f(0/0) = 0$, and $0 f(a/0) = a \lim_{t \to \infty} f(t)/t$.
For a random variable $X$ and a set $A$ we denote the probability that $X$ take a value in $A$ by $\mathbb{P}(X \in A)$, the expectation of the random variable by $\mathbb{E}X$ and the variance by $\Var(X) \coloneqq \mathbb{E} | X - \mathbb{E}X|^2$.
For a probability measure $\mu$ satisfying $\mu(A) = \mathbb{P}(X \in A)$ for all Borel $A$, we write $X \sim \mu$, and when there exists a probability density function such that 
$\mathbb{P}(X \in A) = \int_A f(x) d\gamma(x)$ for a reference measure $\gamma$, we write $X \sim f$.  For a probability measure $\mu$ on $\mathcal{X}$, and an $L^2$ function $f: \mathcal{X} \to \mathbb{R}$,
  we denote $\Var_\mu(f) \coloneqq \Var(f(X))$ for $X \sim \mu$.
%We use $V$ to denote the total variation distance between two densities $p_1$ and $p_2$,
%  $V = V(p_1,p_2)   = \int (p_1 - p_2)_+ d\mu$ where $x_+ = x \wedge 0$. 
 
%  \int_\mathcal{X} \left(f(x) - \int_{\mathcal{X}} f(y) d\mu(y) \right)^2 d\mu(x)$. 
% When $X$ is a random variable with distribution $\mu$, then $Var_\mu(f)$ is just the usual variance of the random variable $f(X)$.

\section{Strongly convex divergences} \label{sec: Strongly convex divergences}

%In this paper we will be especially interested in the consequences of $f$-divergences that satisfy stronger convexity properties.  
\begin{defn} \label{defn: kappa convexity}
   A $\mathbb{R} \cup \{ \infty \}$-valued function $f$ on a convex set $K \subseteq \mathbb{R}$ is $\kappa$-convex when $x,y \in K$ and $t \in [0,1]$ implies
   \begin{align} \label{eq: kappa convexity}
       f((1-t)x+ty) \leq (1-t)f(x) + t f(y) - \kappa t(1-t)(x-y)^2/2.
   \end{align}
\end{defn}
For example, when $f$ is twice differentiable, \eqref{eq: kappa convexity} is equivalent to $f''(x) \geq \kappa$ for $x \in K$. 
 Note that the case $\kappa = 0$ is just usual convexity.

 \begin{prop} \label{eq: strongly convex equivalences}
    For $f: K \to \mathbb{R} \cup \{\infty\}$, and $\kappa \in [0,\infty)$ following are equivalent:
    \begin{enumerate}
        \item $f$ is $\kappa$-convex.
        \item The function $f - \kappa (t-a)^2/2$ is convex for any $a \in \mathbb{R}$
        \item The right handed derivative, defined  as 
            $f'_+(t) \coloneqq \lim_{h \downarrow 0} \frac{f(t+h) - f(t)}{h}$ satisfies,
            \begin{align*}
                f'_+(t) \geq f'_+(s) + \kappa (t-s)
            \end{align*}
            for $t \geq s$.
    \end{enumerate}
 \end{prop}

 \begin{proof}
    Observe that it is enough to prove the result when $\kappa = 0$, where the proposition is reduced to the classical result for convex functions.
 \end{proof}

\begin{defn} \label{defn: kappa convex f divergence}
    An $f$ divergence $D_f( \cdot || \cdot )$ is $\kappa$-convex on an interval $K$ for $\kappa \geq 0$ 
    when the function $f$ is $\kappa$-convex on $K$.
\end{defn}

The table below lists some $\kappa$-convex $f$-divergences of interest to this article. 
\begin{center}
    \begin{tabular}{ |c|c|c|c| } 
    \hline
    Divergence & $f$ & $\kappa$ & Domain \\
     \hline
     \hline
     relative entropy (KL) & $ t \log t $ & $\frac 1 M$ & $(0,M]$ \\
     \hline
     total variation & $\frac{|t-1|}{2}$ & $0$ & $(0,\infty)$ \\
     \hline
     Pearson's $\chi^2$& $(t-1)^2$ & $2$ & $(0,\infty)$\\ 
     \hline
     squared Hellinger & $2( 1 - \sqrt{t} ) $& $M^{- \frac 3 2}/2$ & $(0,M]$\\ 
     \hline
     reverse relative entropy & $-\log t$ & $1/M^2$ & $(0,M]$\\ 
     \hline
     Vincze- Le Cam & $\frac{(t-1)^2}{t+1}$ & $\frac{8}{ (M+1)^3}$ & $(0, M]$ \\
     \hline
     Jensen-Shannon & $(t+1)\log \frac{2}{t+1} + t \log t$ & $\frac{1}{M(M+1)}$ & $(0,M]$ \\
     \hline
     Neyman's $\chi^2$ & $\frac 1 t - 1$ & $2/M^3$ & $(0,M]$ \\
     \hline
     Sason's $s$ & $\log(s+t)^{(s+t)^2 } -  \log(s+1)^{(s+1)^2}$ & $2 \log(s + M) + 3$ & $[M,\infty)$, $s > e^{-3/2}$\\
     \hline
     $\alpha$-divergence & $\frac{4\left( 1- t^{\frac{1 + \alpha}{2}}\right)}{1-\alpha^2}, \hspace{2mm} \alpha \neq \pm 1$ & $M^{\frac{\alpha - 3}{2}}$ & $ \begin{cases}
         [M,\infty), \hspace{2mm} \alpha > 3 \\
            (0, M], \hspace{2mm} \alpha < 3
     \end{cases}$ \\  
     \hline
    \end{tabular}
    \end{center}
    Observe that we have taken the normalization convention on the total variation, which we denote by $|P-Q|_{TV}$, such that $|P- Q|_{TV} = \sup_A |P(A) - Q(A)| \leq 1$.
    Also note,  the $\alpha$-divergence interpolates Pearson's $\chi^2$-divergence when $\alpha = 3$, one half Neyman's $\chi^2$-divergence when 
    $\alpha = -3$, the squared Hellinger divergence when $\alpha = 0$, and has limiting cases, the relative entropy when $\alpha = 1$ and the reverse 
    relative entropy when $ \alpha = -1$.  If $f$ is $\kappa$-convex on $[a,b]$, then its dual divergence $f^*(x) \coloneqq  x f(x^{-1})$ is $\kappa a^3$-convex 
    on $[\frac 1 b, \frac 1 a]$.  Recall that $f^*$ satisfies the equality $D_{f^*}(P||Q) = D_f(Q||P)$.  For brevity, we will use $\chi^2$-divergence to refer to the Pearson 
    $\chi^2$-divergence, and will articulate Neyman's $\chi^2$ explicitly when necessary.

The next lemma is a restatement of Jensen's inequality.

\begin{lem} \label{lem: kappa jensens}
If $f$ is $\kappa$-convex on the range of $X$, 
\[
    \mathbb{E}f(X) \geq f(\mathbb{E}(X)) + \frac \kappa 2 \Var(X).
\]
\end{lem}
\begin{proof}
Apply Jensen's inequality to $f(x) - \kappa x^2/2$.
\end{proof}

% The class of $f$-divergences enjoy many properties, one of which is their joint convexity.  That is for measures $P_i$ and $Q_i$, and $t \in (0,1)$,
% \begin{align}
%     D_f \left( (1-t) P_1 + t P_2 || (1-t) Q_1 + t Q_2 \right) \leq (1-t) D_f(P_1||Q_1) + t D_f(P_2||Q_2). 
% \end{align}
% This can be proven in a straight forward manner.  If we write $\lambda = \frac{t Q_2}{(1-t)Q_1 + tQ_2}$,
% \begin{align}
%     D_f\bigg( (1-t) P_1 + t P_2 &\bigg|\bigg| (1-t) Q_1 + t Q_2 \bigg)
%         \\
%         &=
%             \int f\left( \frac{ (1-t)P_1 + t P_2 }{(1-t) Q_1 + t Q_2 } \right) d( (1-t)Q_1 + tQ_2)
%                 \\
%         &=
%             \int f \left((1-\lambda) \frac{P_1}{Q_1} + \lambda \frac{P_2}{Q_2} \right) d((1-t)Q_1 + tQ_2)
%                 \\
%         &\leq
%             \int \left((1-\lambda) f\left(\frac{P_1}{Q_1}\right) + \lambda f\left(\frac{P_2}{Q_2}\right) \right) d((1-t)Q_1 + tQ_2)
%                 \\
%         &=
%             (1-t) D_f(P_1||Q_1) + t D_f(P_2 ||Q_2).
% \end{align}
% We will revisit this inequality in the case that $f$ satisfies stronger convexity properties.

%\section{f-divergences}
%\subsection{Preliminaries}
For a convex function $f$ such that $f(1) = 0$, and $c \in \mathbb{R}^d$ 
the function $\tilde{f}(t) = f(t) + c(t-1)$ remains a convex function, and what is more satisfies
\begin{align*}
    D_f(P||Q) = D_{\tilde{f}}(P ||Q)
\end{align*}
since $\int c ( p/q- 1) q d \mu = 0$.

%\subsection{Strongly convex divergences}

\begin{defn}[$\chi^2$-divergence] \label{def: chi square divergence}
     For $f(t) = (t-1)^2$, we write
    \begin{align*}
        \chi^2(P || Q) \coloneqq D_f(P ||Q)
    \end{align*}
\end{defn}

The following result shows that every strongly convex divergence can be lower bounded,
up to its convexity constant $\kappa >0$, by the $\chi^2$-divergence.

\begin{thm} \label{thm: chi square divergence is smallest strongly convex divergence}
    For a $\kappa$-convex function $f$,
    \begin{align*}
        D_f(P || Q) \geq \frac{\kappa}{2}\chi^2(P||Q).
    \end{align*}
\end{thm}

\begin{proof}
    Define a $\tilde{f}(t) = f(t) - f'_+(1)(t-1)$, and note that $\tilde{f}$ 
    defines the same $\kappa$-convex divergence as 
    $f$.  So we may assume without loss of generality that $f'_+$ is uniquely zero 
    when $t = 1$.  Since $f$ is $\kappa$-convex $\phi: t \mapsto f(t) - \kappa (t-1)^2 /2$
    is convex, and by $f'_+(1)  =0$,  $\phi'_+(1) = 0$ as well.  Thus $\phi$ takes
    its minimum when $t = 1$ and hence $\phi \geq 0$ so that $f \geq \kappa(t-1)^2/2$.
    Computing,
    \begin{align*}
        D_f(P||Q) 
            &=
                \int f \left( \frac{dP}{dQ} \right) d Q
                    \\
            &\geq 
                \frac{\kappa}{2} \int \left( \frac{dP}{dQ} - 1\right)^2 d Q
                    \\
            &=
                \frac{\kappa}{2} \chi^2(P||Q).
    \end{align*}
\end{proof}

The above proof uses a pointwise inequality between convex functions to derive
an inequality between their respective divergences.  This simple technique was shown
to have useful implications by Sason and Verd\'u in \cite{sason2016f}, where it appears 
as Theorem 1, and  was
used to give sharp comparisons in several $f$-divergence inequalities.

\begin{thm}[Sason-Verd\'u \cite{sason2016f}] \label{thm: Sason-Verdu functional dominance in f divergences}
    For divergences defined by $g$ and $f$ with 
    $c f(t) \geq g(t)$ for all $t$, then
    \begin{align*}
        D_g(P||Q) \leq c D_f(P||Q).
    \end{align*}
    Morever if $f'(1) = g'(1) = 0$ then
    \begin{align*}
        \sup_{P \neq Q} \frac{D_g(P||Q)}{D_f(P||Q)}  
            = 
                \sup_{t \neq 1} \frac{g(t)}{f(t)}.
    \end{align*}
\end{thm}

\begin{cor}
    For a smooth $\kappa$-convex divergence $f$, the inequality
    \begin{align}   
        D_f(P ||Q)  \geq  \frac{\kappa}{2}\chi^2(P||Q)
    \end{align}
    is sharp multiplicatively in the sense that 
    \begin{align} \label{eq: if -  sharpness in chi square inequality}
        \inf_{P \neq Q} \frac{ D_f(P||Q)}{\chi^2(P||Q)} = \frac{\kappa}{2}.
    \end{align}
     if $f''(1) =\kappa$.% and sharp additively in the sense that
    % \begin{align} \label{eq: only if - sharpness in chi square inequality}
    %     \inf_{P \neq Q}  D_f(P ||Q) - \frac{\kappa}{2} \chi^2(P||Q)= 0
    % \end{align}
    %  only if $f''(1) = \kappa$.
\end{cor}

\begin{proof}
    Without loss of generality we assume that $f'(1) = 0$.
    If $f''(1)  = \kappa + 2\varepsilon$ for some $\varepsilon >0$, then taking $g(t)= (t-1)^2$ and applying Theorem 
    \ref{thm: Sason-Verdu functional dominance in f divergences} and Theorem \ref{thm: chi square divergence is smallest strongly convex divergence}
    \begin{align}   
        \sup_{P \neq Q} \frac{D_g(P||Q)}{D_f(P||Q)}  
        = 
            \sup_{t \neq 1} \frac{g(t)}{f(t)} \leq \frac{2}{\kappa}.
    \end{align}
    Observe that after two applications of L'Hospital,
    \begin{align*}
        \lim_{\varepsilon \to 0} \frac{g(1 + \varepsilon)}{f(1+\varepsilon)} = \lim_{\varepsilon \to 0} \frac{g'(1+\varepsilon)}{f'(1+\varepsilon)} =  \frac{g''(1)}{f''(1)}
             = \frac{2}{\kappa} \leq \sup _{t \neq 1} \frac{g(t)}{f(t)}.
    \end{align*}
     Thus \eqref{eq: if -  sharpness in chi square inequality} follows.
    % Now suppose that $f''(1) = 1 + 2 \varepsilon$ for $\varepsilon >0$.  By the continuity of $f$ and $f''$, there exists $\delta >0$ such that for $|x-1| \leq \delta$ implies $f(x)/g(x) \geq \frac{\kappa + \varepsilon}{2}$ and
    % $f''(x) \geq \kappa + \varepsilon$.  Consider $|t| > \delta$ and the Taylor expansion bounds,
    % \begin{align*}
    %     f(1+t) 
    %         &= 
    %             f(1) + f'(1) t + \int_0^t \int_0^s f''(w) dw ds
    %                 \\
    %         &=
    %             \int_0^\delta \int_0^s f''(w) dw ds + \int_\delta^t \int_0^s f''(w) dw ds
    %                 \\
    %         &\geq 
    %             \int_0^\delta \int_0^s (\kappa +\varepsilon) dw ds + \int_\delta^t \int_0^s \kappa \hspace{1.5mm} dw ds
    %                 \\
    %         &= 
    %             \frac{\kappa}{2} t^2 + \frac{\varepsilon \delta^2}{2}
    % \end{align*}
    % Thus $f(t) \geq \frac{\kappa}{2} g(t) + \frac{\varepsilon \delta^2}{2}$ so that after integrating,
    % \begin{align} 
    %     D_f(P||Q) - \frac{\kappa}{2} \chi^2(P||Q) \geq \frac{\varepsilon \delta^2}{2} >0
    % \end{align}
\end{proof}

\begin{prop} \label{prop: kappa convexity of kappa convex divergences continuous}
    When $D_f(\cdot || \cdot)$ is an $f$ divergence such that $f$ is $\kappa$-convex on  $[a,b]$ and 
    that $P_\theta$  and $Q_\theta$ are probability measures indexed by a set $\Theta$ 
    such that $a \leq \frac{dP_\theta}{dQ_\theta}(x) \leq b$,
    holds for all $\theta$ and $P \coloneqq \int_\Theta P_\theta d \mu(\theta)$ and $Q \coloneqq \int_\Theta Q_\theta d \mu(\theta)$ 
    for a probability measure $\mu$ on $\Theta$, then
    \begin{align}
        D_f( P &||Q) 
            \leq  \int_\Theta D_f(P_\theta || Q_\theta) d\mu(\theta)  - \frac{\kappa}{2} \int_\Theta \int_{\mathcal{X}} \left( \frac{dP_\theta}{dQ_\theta} -
            \frac{dP}{dQ}  \right)^2 dQ d\mu,
    \end{align}

    In particular when $Q_\theta = Q$ for all $\theta$
    \begin{align}
        D_f( P &||Q) \\
        &\leq  \int_\Theta  D_f(P_\theta || Q) d\mu(\theta) -  \frac{\kappa}{2} \int_\Theta \int_\mathcal{X} \left( \frac{d P_\theta}{dQ} - \frac{dP}{dQ} \right)^2 dQ d\mu(\theta)
            \\
            &\leq
            \int_\Theta  D_f(P_\theta || Q) d\mu(\theta) - \kappa \int_\Theta | P_\theta - P |^{2}_{TV} d \mu(\theta)
    \end{align}
    \end{prop}

    \begin{proof}
        Let $d\theta$ denote a reference measure dominating $\mu$ so that $d \mu = \varphi(\theta) d \theta$ then write
        $\nu_\theta = \nu(\theta,x) = \frac{ d Q_\theta}{dQ}(x) \varphi(\theta)$.
        \begin{align}
            D_f (P ||Q )
                &=
                    \int_{\mathcal{X}} f\left( \frac{dP}{dQ} \right) dQ
                        \\
                &=
                    \int_{\mathcal{X}} f\left( \int_\Theta \frac{dP_\theta}{dQ} d\mu(\theta) \right)  dQ
                        \\
                &=
                    \int_{\mathcal{X}} f\left( \int_\Theta \frac{dP_\theta}{dQ_\theta} \nu(\theta,x) d\theta \right)  dQ
        \end{align}
        By Jensen's inequality, as in Lemma \ref{lem: kappa jensens}
        \begin{align*}
            f\left( \int_\Theta \frac{dP_\theta}{dQ_\theta} \nu_\theta d\theta \right)
                \leq
                    \int_\theta f &\left(\frac{dP_\theta}{dQ_\theta} \right) \nu_\theta d \theta -
                        \frac{\kappa}{2} \int_\Theta \left( \frac{dP_\theta}{dQ_\theta} -
                         \int_\Theta \frac{dP_\theta}{dQ_\theta} \nu_\theta d\theta \right)^2 \nu_\theta d\theta
        \end{align*}
        Integrating this inequality gives
        \begin{align} \label{eq: inequality we want}
            D_f(P||Q) \leq \int_{\mathcal{X}} \left( \int_\theta f \left(\frac{dP_\theta}{dQ_\theta} \right) \nu_\theta d \theta -
            \frac{\kappa}{2} \int_\Theta \left( \frac{dP_\theta}{dQ_\theta} -
             \int_\Theta \frac{dP_\theta}{dQ_\theta} \nu_\theta d\theta \right)^2 \nu_\theta d\theta \right) dQ
        \end{align}
            Note that
        \begin{align*}
           \int_{\mathcal{X}} \int_\Theta \left( \frac{dP_\theta}{dQ_\theta} dQ-
             \int_\Theta \frac{dP_\theta}{dQ_{\theta_0}} \nu_{\theta_0} d{\theta_0} \right)^2 \nu_\theta d\theta dQ
                = 
                \int_\Theta \int_{\mathcal{X}} \left( \frac{dP_\theta}{dQ_\theta} -
                 \frac{dP}{dQ}  \right)^2 dQ d\mu,
        \end{align*}
            and
        \begin{align}
                    \int_{\mathcal{X}} \int_\Theta f\left(  \frac{dP_\theta}{dQ_\theta} \right) \nu(\theta,x) d\theta  dQ
                &=
                    \int_\Theta \int_{\mathcal{X}}  f\left(  \frac{dP_\theta}{dQ_\theta} \right) \nu(\theta,x) dQ d\theta  
                        \\
                &=
                    \int_\Theta \int_{\mathcal{X}}  f\left(  \frac{dP_\theta}{dQ_\theta} \right)  dQ_\theta d\mu(\theta) 
                        \\
                &=
                    \int_\Theta D(P_\theta || Q_\theta) d\mu(\theta) 
        \end{align}
            Inserting these equalities into \eqref{eq: inequality we want} gives the result. \\
            To obtain the total variation bound one needs only to apply Jensen's inequality,
            \begin{align}
                \int_\mathcal{X} \left( \frac{d P_\theta}{dQ} - \frac{dP}{dQ} \right)^2 dQ
                    &\geq
                        \left(\int_\mathcal{X} \left| \frac{d P_\theta}{dQ} - \frac{dP}{dQ} \right| dQ \right)^2
                            \\
                    &=
                        |P_\theta - P|^2_{TV}.
            \end{align}
    \end{proof}
Observe that taking $Q = P = \int_\Theta  P_\theta d \mu(\theta)$ in Proposition \ref{prop: kappa convexity of kappa convex divergences continuous}, one obtains a lower bound for the average $f$-divergence from the set of distribution to their barycenter,
by the mean square total variation of the set of distributions to the barycenter,
\begin{align} \label{eq: Reverse Pinsker for Barycenter}
    \kappa \int_{\Theta} |P_\theta - P|^2_{TV} d\mu(\theta) \leq \int_{\Theta} D_f(P_\theta|| P) d\mu(\theta).
\end{align}

    The next result shows that for $f$ strongly convex, Pinsker type inequalities can never be reversed, 
\begin{prop} \label{prop: no reverse pinskers for strongly convex}
    Given $f$ strongly convex and $M >0$, there exists $P$, $Q$ measures such that
    \begin{align}
        D_f(P||Q) \geq M |P-Q|_{TV}.
    \end{align}
\end{prop}

\begin{proof}
By $\kappa$-convexity $\phi(t) = f(t) - \kappa t^2/2$ is a convex function.  Thus $\phi(t) \geq \phi(1) + \phi'_+(1)(t-1) = (f'_+(1) - \kappa)(t-1)$ and hence
$
   \lim_{t \to \infty} \frac{f(t)}{t} \geq \lim_{t \to \infty} \kappa t/2 + (f'_+(1) - \kappa)\left(1-\frac{1}{t} \right) = \infty.
$
Taking measures on the two points space $P = \{1/2,1/2\}$ and $Q = \{ 1/2t, 1- 1/2t\}$ gives
$
    D_f(P||Q) \geq \frac 1 2 \frac{f(t)}{t} 
$
which tends to infinity with $t \to \infty$, while $|P-Q|_{TV}\leq 1$.

\end{proof}

In fact, building on the work of \cite{basu2011statistical, LV06:1},  Sason and Verdu proved in \cite{sason2016f}, that for any $f$ divergence, $\sup_{P \neq Q} \frac{D_f(P||Q)}{|P-Q|_{TV}} = f(0) + f^*(0)$. 
Thus, an $f$-divergence can be bounded above by a constant multiple of a the total variation, if and only if $f(0) + f^*(0) < \infty$.  From this perspective, Proposition \ref{prop: no reverse pinskers for strongly convex}
is simply the obvious fact that strongly convex functions have super linear (at least quadratic) growth at infinity.

\section{Skew divergences} \label{sec: Skew divergences}
If we denote $Cvx(0,\infty)$ to be quotient of the cone of convex functions $f$ on $(0,\infty)$ such that $f(1) = 0$ under the equivalence relation $f_1 \sim f_2$ 
when $f_1- f_2 = c(x-1)$ for $c \in \mathbb{R}$, then the map $f \mapsto D_f(\cdot|| \cdot )$ gives 
a linear isomorphism between $Cvx(0,\infty)$ and the space of all $f$-divergences.  The mapping $\mathcal{T}: Cvx(0,\infty) \to Cvx(0,\infty)$ defined
by $\mathcal{T}f = f^*$, where we recall $f^*(t) = t f(t^{-1})$, gives an involution of $Cvx(0,\infty)$.  Indeed, $D_{\mathcal{T}f}(P||Q) = D_f(Q||P)$,
so that $D_{\mathcal{T}(\mathcal{T}(f))}(P||Q) = D_f(P||Q)$.  Mathematically, skew divergences give an interpolation of this involution as 
$$(P,Q) \mapsto D_f((1-t)P + t Q || (1-s)P + s Q)$$ gives $D_f(P||Q)$ by taking $s = 1$ and $t = 0$ or yields $D_{f^*}(P||Q)$ by taking
$s=0$ and $t= 1$.  

Moreover as mentioned in the introduction, skewing imposes boundedness of the Radon-Nikodym derivative $\frac{dP}{dQ}$, which
allows us to constrain the domain of $f$-divergences and leverage $\kappa$-convexity to obtain $f$-divergence inequalities in this
section.

The following appears as Theorem III.1 in the preprint \cite{melbourne2020differential}.
It states that skewing an $f$-divergence preserves its status as such. This guarantees that the generalized skew divergences of this section are
indeed $f$-divergences. A proof is given in the
appendix for the convenience of the reader. 

\begin{thm}[Melbourne et al \cite{melbourne2020differential}] \label{thm: skewing preserves f-divergence}
    For $t,s \in [0,1]$ and an $f$-divergence, $D_f(\cdot || \cdot )$, in the sense that
    \begin{align}   
        S_f(P||Q) \coloneqq D_f( (1-t) P + tQ || (1-s)P + sQ)
    \end{align}
    is an $f$-divergence if $D_f$ is.
\end{thm}

\begin{defn} \label{def: Symmetrization divergence}
    For an $f$-divergence, its skew symmetrization,
    \begin{align*}
        \Delta_f(P||Q) 
            \coloneqq 
                \frac 1 2 D_f \left( P \bigg|\bigg|\frac{P +Q}{2} \right) +
                     \frac 1 2 D_f\left(Q \bigg|\bigg|\frac{P+Q}{2} \right).
    \end{align*} 
\end{defn}
 $\Delta_f$ is determined by the convex function 
 \begin{align} \label{eq: formula for skew symmetric f}
    x \mapsto \frac{1+x}{2} \left( f \left( \frac{2x}{1+x} \right) + f \left( \frac{2}{1+x} \right) \right).
 \end{align}
 Observe that $\Delta_f(P||Q) = \Delta_f(Q||P)$, 
 and when $f(0)< \infty$, $\Delta_f(P||Q) \leq \sup_{x \in [0,2]} f(x) < \infty$ for all $P,Q$ since $\frac{ dP}{d(P+Q)/2}$, $\frac{ dQ}{d(P+Q)/2} \leq 2$.
When $f(x) = x\log x$, the relative entropy's skew symmetrization is the Jensen-Shannon divergence.  When 
$f(x) = (x-1)^2$ up to a normalization constant the $\chi^2$-divergence's skew symmetrization is the Vincze-Le Cam divergence which we state below for emphasis. See \cite{topsoe2000some} for more
background on this divergence, where it is referred to as the triangular discrimination.  
\begin{defn} \label{def: Triangular discrimination}
    When $f(t) = \frac{(t-1)^2}{t+1}$ denote the Vincze-Le Cam divergence by
    \begin{align*}
        \Delta(P||Q) \coloneqq D_f(P||Q).
    \end{align*}
\end{defn}
If one denotes the skew symmetrization of the $\chi^2$-divergence by $\Delta_{\chi^2}$, one can compute easily from \eqref{eq: formula for skew symmetric f} that $\Delta_{\chi^2}(P||Q) = 
\Delta(P||Q)/2$.    We note that although skewing preserves $0$-conexity, by the above example, it does not preserve $\kappa$-convexity in general. The skew symmetrization of the $\chi^2$-divergence a $2$-convex divergence while $f(t) = (t-1)^2/(t+1)$ corresponding to the
Vincze-Le Cam divergence satisfies $f''(t) = \frac 8 {(t+1)^3}$, which
cannot be bounded away from zero on $(0,\infty)$.

\begin{cor} \label{cor: Triangular discrimination as smallest symmetrized}
    For an $f$-divergence such that $f$ is a $\kappa$-convex on $(0,2)$, 
    \begin{align}  \label{eq: chi square give smallest skew symmetrized divergence}
        \Delta_f(P||Q) \geq  \frac \kappa 4 \Delta(P||Q) = \frac \kappa 2 \Delta_{\chi^2}(P||Q),
    \end{align}
    with equality when the $f(t) = (t-1)^2$ corresponding the the $\chi^2$-divergence,
    where $\Delta_f$ denotes the skew symmetrized divergence associated to $f$ and $\Delta$ is the Vincze- Le Cam divergence.
\end{cor}
%Note that \eqref{eq: chi square give smallest skew symmetrized divergence} is an equality when 
\begin{proof}
    Applying Proposition \ref{prop: kappa convexity of kappa convex divergences continuous}
    \begin{align*}
        0 
            &=
                D_f\left( \frac{P+Q}{2}\bigg|\bigg| \frac{Q+P}{2} \right)
                    \\
            &\leq
                \frac{1}{2} D_f\left( P\bigg|\bigg| \frac{Q+P}{2} \right)+ \frac 1 2D_f\left( Q \bigg|\bigg| \frac{Q+P}{2} \right) 
                    - \frac{\kappa}{8}  \int \left( \frac{2P}{P+Q} - \frac{2Q}{P + Q} \right)^2 d(P+Q)/2
                        \\
            &=
                \Delta_f(P ||Q) - \frac{\kappa}{4} \
                \Delta(P||Q).
    \end{align*}
\end{proof}
When $f(x) = x \log x$, we have $f''(x) \geq \frac{\log e} 2$ on $[0,2]$, which demonstrates that up to a constant $\frac {\log e} 8$ the Jensen-Shannon
divergence bounds the Vincze-Le Cam divergence. See \cite{topsoe2000some} for improvement of the inequality  in the case of the Jensen-Shannon divergence, called the ``capacitory discrimination''
in the reference, by a factor of $2$.

% As another example, taking $f(x) = 2(1 - \sqrt{x})$, the squared Hellinger distance, $\mathcal{H}_{\frac 1 2}(P||Q) = \int f(dP/dQ) dQ$. on $[0,M]$ then $f''(x) \geq \kappa \coloneqq \frac{1} 2 M^{- \frac 3 2} $. 
%  In particular on
% $[0,2]$, $f''(x) \geq \kappa = 2^{-5/2}$. Thus the skew-symmetrized Hellinger distance denoted by $\Delta_{\mathcal{H}}$, is always finite, preserves symmetry,
%  and by Corollary \ref{cor: Triangular discrimination as smallest symmetrized}
% is bounded below by the triangular discrimination,
% \begin{align*}
%     \Delta_{\mathcal{H}}(P||Q) 
%         &\geq 
%             \frac 1 2 \mathcal{H}(P|| P/2+Q/2) + \frac 1 2 \mathcal{H}(Q|| P/2 + Q/2)
%                 \\
%         &=
%             2^{-9/2} \Delta(P||Q).
% \end{align*}

We will now investigate more general, non-symmetric skewing in what follows.

\begin{prop} \label{prop: generalized Jensen-Shannon bound}
    For $\alpha, \beta \in [0,1]$, define 
    \begin{align}
        C(\alpha) \coloneqq \begin{cases} 1-\alpha & \hbox{ when }  \alpha \leq \beta\\ \alpha & \mbox{ when } \alpha > \beta,
        \end{cases}
    \end{align}
    and
    \begin{align}
        S_{\alpha,\beta}(P || Q) \coloneqq  D((1-\alpha)P+ \alpha Q|| (1-\beta) P + \beta Q).
    \end{align}
    Then
    \begin{align}
        S_{\alpha,\beta}(P ||Q) \leq C(\alpha) D_\infty(\alpha||\beta) | P - Q |_{TV} 
    \end{align}
\end{prop}

We will need the following lemma originally proved by Audenart in the quantum setting 
\cite{audenaert2014quantum}.  It is based on 
a diffential relationship between the skew divergence \cite{lee1999measures} and the \cite{gyorfi2001class}, see \cite{melbourne2019relationships,nishiyama2020relations}.
\begin{lem} [Theorem III.1 \cite{melbourne2020differential}] \label{lem: our lemma on skew divergence bounded by total variation}
   For $P$ and $Q$ probability measures, and $t \in [0,1]$
       \begin{align}
           S_{0,t}(P||Q) \leq -\log t | P - Q |_{TV}.
       \end{align}
\end{lem}

\begin{proof}[Proof of Theorem \ref{prop: generalized Jensen-Shannon bound}]
   If $\alpha \leq \beta$, then $D_\infty( \alpha || \beta) = \log \frac{1-\alpha}{1-\beta}$ and $C(\alpha) = 1-\alpha$.  Also,
   \begin{align}
       (1-\beta) P + \beta Q = t \left( (1-\alpha) P + \alpha Q\right) + (1-t) Q
   \end{align}
   with $t = \frac{1-\beta}{1-\alpha}$, thus
   \begin{align}
       S_{\alpha,\beta}(P||Q) 
           &=
               S_{0,t}((1-\alpha) P + \alpha Q || Q)
                   \\
           &\leq 
               - \log t | ((1-\alpha) P + \alpha Q) - Q |_{TV}
                   \\
           &=
               C(\alpha) D_\infty(\alpha|| \beta) |P - Q|_{TV},
   \end{align}
    where the inequality follows from Lemma \ref{lem: our lemma on skew divergence bounded by total variation}. 
    Following the same argument for $\alpha > \beta$, so that $C(\alpha) = \alpha$, $D_\infty(\alpha || \beta) = \log \frac \alpha \beta$, and
    \begin{align}
        (1-\beta) P + \beta Q = t \left( (1-\alpha) P + \alpha Q\right) + (1-t) P
    \end{align}
       for $t = \frac{\beta}{\alpha}$ completes the proof.  Indeed,
       \begin{align}
           S_{\alpha,\beta}(P||Q) 
           &=
               S_{0,t}((1-\alpha) P + \alpha Q || P)
                   \\
           &\leq 
               - \log t | ((1-\alpha) P + \alpha Q) - P |_{TV}
                   \\
           &=
               C(\alpha) D_\infty(\alpha|| \beta) |P - Q|_{TV}.
       \end{align}
\end{proof}

We recover the classical bound \cite{lin1991divergence, topsoe2000some} of the Jensen-Shannon divergence by the total variation.

\begin{cor}
    For probability measure $P$ and $Q$,
    \begin{align}
        JSD(P||Q) \leq \log 2 |P-Q|_{TV}
    \end{align}
\end{cor}
\begin{proof}
    Since $JSD(P||Q) = \frac 1 2 S_{0,\frac 1 2}(P ||Q) + \frac 1 2 S_{1, \frac 1 2}(P||Q)$
\end{proof}

Proposition \ref{prop: generalized Jensen-Shannon bound} gives a sharpening of Lemma 1 of Neilsen \cite{nielsen2020generalization} 
who proved $S_{\alpha,\beta}(P || Q) \leq D_\infty(\alpha || \beta)$, 
and used the result to establish the boundedness of a generalization of the Jensen-Shannon Divergence.  

\begin{defn}[Nielsen \cite{nielsen2020generalization}]
    For $p$ and $q$ densities with respect to a reference measure $\mu$, $w_i > 0$, such that $\sum_{i=1}^n w_i = 1$ and $\alpha_i \in [0,1]$, define
    \begin{align}
        JS^{\alpha,w}(p : q) = \sum_{i=1}^n w_i D( (1-\alpha_i) p + \alpha_i q || (1- \bar{\alpha} ) p + \bar{\alpha} q)
    \end{align}
    where $\sum_{i=1}^n w_i \alpha_i = \bar{\alpha}$.
\end{defn}
Note that when $n=2$, $\alpha_1 = 1$, $\alpha_2 = 0$ and $w_i = \frac 1 2$ that $JS^{\alpha,w}(p : q) = JSD(p||q)$, the usual Jensen-Shannon divergence.
We now demonstrate that Neilsen's generalized Jensen-Shannon Divergence can be bounded by the total variation distance just as the ordinary Jensen-Shannon Divergence.

\begin{thm}\label{thm: generalized Jensen Shannon total variation bound}
   For $p$ and $q$ densities with respect to a reference measure $\mu$,  $w_i > 0$, such that $\sum_{i=1}^n w_i = 1$ and $\alpha_i \in (0,1)$ then,
   \begin{align}
       \log e \Var_w(\alpha) | p - q|^2_{TV} \leq JS^{\alpha,w}(p : q) \leq \mathcal{A} H(w) | p - q |_{TV}
   \end{align}
   where $H(w) \coloneqq -\sum_i w_i \log w_i$ and $\mathcal{A} = \max_i |\alpha_i - \bar{\alpha}_i |$ with $\bar{\alpha}_i = \sum_{j \neq i} \frac{w_j \alpha_j}{1- w_i}$
\end{thm}

Note that since $\bar{\alpha}_i$ is the $w$ average of the $\alpha_j$ terms with $\alpha_i$ removed, $\bar{\alpha}_i \in [0,1]$ and thus $\mathcal{A} \leq 1$.  We will 
need the following Theorem from \cite{melbourne2020differential} for the upper bound.
\begin{thm}[\cite{melbourne2020differential} Theorem 1.1] \label{thm: concavity deficit}
    For $f_i$ densities with respect to a common reference measure $\gamma$, and $\lambda_i > 0$ such that $\sum_{i=1}^n \lambda_i = 1$,
    \begin{align}
        h_\gamma(\sum_i \lambda_i f_i) - \sum_i \lambda_i h_\gamma(f_i) \leq \mathcal{T} H(\lambda),
    \end{align}
    where $h_\gamma(f_i) \coloneqq - \int f_i(x) \log f_i(x) d\gamma(x)$, and $\mathcal{T} = \sup_i | f_i - \tilde{f}_i |_{TV}$ with $\tilde{f}_i = \sum_{j \neq i} \frac{\lambda_j }{1 -\lambda_i } f_j$.
\end{thm}

\begin{proof}[Proof of Theorem \ref{thm: generalized Jensen Shannon total variation bound}]
We apply Theorem \ref{thm: concavity deficit} with $f_i = (1-\alpha_i) p + \alpha_i q$, $\lambda_i = w_i$, and noticing that in general
\begin{align}
    h_\gamma(\sum_i \lambda_i f_i) - \sum_i \lambda h_\gamma(f_i) = \sum_i \lambda_i D( f_i || f),
\end{align}
we have
\begin{align}
    JS^{\alpha,w}(p : q) 
       &= 
           \sum_{i=1}^n w_i D( (1-\alpha_i) p + \alpha_i q || (1- \bar{\alpha} ) p + \bar{\alpha} q)
               \\
       &\leq
           \mathcal{T} H(w).
\end{align}
It remains to determine $\mathcal{T} = \max_i | f_i - \tilde{f}_i |_{TV}$,
\begin{align}
   \tilde{f}_i - f_i
       &=
           \frac{f - f_i}{1-\lambda_i}
               \\
       &=
           \frac{((1- \bar{\alpha})p + \bar{\alpha} q) - ((1-\alpha_i)p + \alpha_i q)}{1 -w_i}
               \\
       &=
           \frac{(\alpha_i - \bar{\alpha})(p-q)}{1 -w_i}
               \\
       &=
       (\alpha_i - \bar{\alpha}_i)(p-q).
\end{align}
Thus
   $\mathcal{T}
       =
           \max_i (\alpha_i - \bar{\alpha}_i) |p-q|_{TV}
       =
           \mathcal{A} |p - q|_{TV},$ and the proof of the upper bound is complete.\\

To prove the lower bound,
we apply Pinsker's inequality, $2 \log e| P-Q|_{TV}^2 \leq  D(P||Q),$
\begin{align}
   JS^{\alpha, w}(p : q)
       &=
           \sum_{i=1}^n w_i D((1-\alpha_i) p + \alpha_i q|| (1-\bar \alpha)p + \bar \alpha q)
               \\
       &\geq
           \frac 1 2 \sum_{i=1}^n w_i 2 \log e |((1-\alpha_i) p + \alpha_i q) - ((1-\bar \alpha)p + \bar \alpha q) |_{TV}^2
               \\
       &=
           \log e \sum_{i=1}^n w_i (\alpha_i - \bar \alpha)^2 | p -q|_{TV}^2
               \\
       &=
           \log e \Var_w(\alpha) | p -q|^2_{TV}.
\end{align}
\end{proof}

\begin{defn}
    Given an $f$-divergence, densities $p$ and $q$ with respect to common reference measure, $\alpha \in [0,1]^n$ and $w \in (0,1)^n$ such that $\sum_i w_i = 1$
    define its generalized skew divergence
        \begin{align}   
            D_f^{\alpha,w}(p:q) = \sum_{i=1}^n w_i D_f((1-\alpha_i) p + \alpha_i q|| (1-\bar \alpha) p + \bar \alpha q).
        \end{align}
        where $\bar \alpha = \sum_i w_i \alpha_i$.
\end{defn}

Note that by Theorem \ref{thm: skewing preserves f-divergence}, $D_f^{\alpha,w}(\cdot|| \cdot )$ is an $f$-divergence.
The generalized skew divergence of the relative entropy is the generalized Jensen-Shannon divergence
 $JS^{\alpha,w}$.  We will denote the generalized skew divergence of the $\chi^2$-divergence from $p$ to $q$
 by 
\begin{align} 
    \chi^2_{\alpha,w}(p:q) \coloneqq \sum_i w_i \chi^2((1-\alpha_i)p + \alpha_i q|| (1-\bar \alpha p + \bar \alpha q)
\end{align}

Note that when $n=2$ and $\alpha_1 = 0$, $\alpha_2 = 1$ and $w_i = \frac 1 2$, we recover the skew symmetrized divergence in 
Definition \ref{def: Symmetrization divergence}
\begin{align}
    D^{(0,1),(1/2,1/2)}_f(p:q) = \Delta_f(p||q)
\end{align}

The following theorem shows that the usual upper bound for the relative entropy by the $\chi^2$-divergence can be reversed up to a factor in the skewed case.
\begin{thm} \label{thm: generalized chi square upper bounded by generalized jensens shannon}
For $p$ and $q$ with a common dominating measure $\mu$,
\begin{align*}
    \chi^2_{\alpha,w}(p: q) \leq N_\infty(\alpha,w) JS^{\alpha,w}(p:q).
\end{align*}
\end{thm}

Writing $N_\infty(\alpha,w) = \max_i \max \left\{ \frac{1- \alpha_i}{1- \bar \alpha}, \frac{\alpha_i}{\bar \alpha} \right\}$. 
For $\alpha \in [0,1]^n$ and $w \in (0,1)^n$ such that $\sum_i w_i = 1$, we will use the notation $N_\infty(\alpha, w) \coloneqq \max_i e^{D_\infty(\alpha_i || \bar \alpha)}$
where $\bar \alpha \coloneqq \sum_i w_i \alpha_i$.
% \begin{cor} \label{cor: generalized jensens shannon lower bounded by total variation squared}
%     For $w, \alpha \in \mathbb{R}^n$ such that $w_i >0$, $\sum_i w_i = 1$ and $\alpha_i \in [0,1]$ then for $p$ and $q$ probability densities,
%     \begin{align*}
%         \frac{\Var_w (\alpha)}{2N_\infty(\alpha)} | p -q |_{TV}^2 \leq   JS^{\alpha, w} ( p : q).
%     \end{align*}
% \end{cor}

\begin{proof}
    By definition,
    \begin{align*}
        JS^{\alpha,w}(p : q) = \sum_{i=1}^n w_i D( (1-\alpha_i) p + \alpha_i q || (1-\bar \alpha) p + \bar{\alpha} q).
    \end{align*}
    Taking $P_i$ to be the measure associated to $(1-\alpha_i)p + \alpha_i q$ and $Q$ given by $(1-\bar \alpha) p + \bar{\alpha} q$,
    then
    \begin{align}
        \frac{dP_i}{dQ} 
            = 
                \frac{(1-\alpha_i)p + \alpha_i q}{(1-\bar \alpha) p + \bar \alpha q} 
            \leq 
                \max \left\{ \frac{1-\alpha_i}{1-\bar \alpha} , \frac{\alpha_i}{\bar \alpha} \right\}
            = e^{D_\infty(\alpha_i||\bar \alpha)} \leq N_\infty(\alpha,w).
    \end{align}
    Since $f(x) = x \log x$, the convex function associated to the usual KL divergence, satisfies $f''(x) = \frac 1 x$, $f$ is $e^{-D_\infty(\alpha)}$-convex 
    on $[0, \sup_{x,i} \frac{dP_i}{dQ}(x)]$, applying Proposition \ref{prop: kappa convexity of kappa convex divergences continuous}, we obtain
    \begin{align} \label{eq: right above us}
        D\left( \sum_i w_i P_i \bigg|\bigg| Q\right) \leq \sum_i w_i D(P_i || Q) - \frac{\sum_i w_i \int_\mathcal{X} \left( \frac{d P_i}{dQ} - \frac{dP}{dQ} \right)^2 dQ}{2 N_\infty(\alpha,w)}.
    \end{align}
    Since $Q = \sum_i w_i P_i$, the left hand side of \eqref{eq: right above us} is zero, while
    \begin{align}
        \sum_i w_i \int_{\mathcal{X}} \left( \frac{d P_i}{dQ} - \frac{dP}{dQ} \right)^2 dQ 
            &= 
                \sum_i w_i \int_{\mathcal{X}} \left( \frac{d P_i}{dP} - 1 \right)^2 dP
                    \\
            &= 
                \sum_i w_i \chi^2 (P_i || P)
                    \\
            &=
                \chi^2_{\alpha,w} (p : q).
    \end{align}  
    Rearranging gives,
    \begin{align}
        \frac{\chi^2_{\alpha,w}(p:q)}{2 N_\infty(\alpha,w)} \leq JS^{\alpha,w}(p:q),
    \end{align}
    % Inserting the equality
    % \begin{align}
    %     |P_i - Q |_{TV}^2  = | (\alpha_i - \bar \alpha) (p - q) |_{TV}^2 = (\alpha_i - \bar \alpha)^2 |p-q|_{TV}^2,
    % \end{align}
    % gives
    % \begin{align}
    %     \frac{\Var (\alpha) |p - q|_{TV}^2}{2 N_\infty(\alpha,w)}  \leq JS^{\alpha,w}(p:q),
    % \end{align}
    which is our conclusion.
\end{proof}

\section{Total Variation Bounds and Bayes risk} \label{sec: TV bounds and Bayes risk}

In this section we will derive bounds on the Bayes risk associated to a family of probability 
measures with a prior distribution $\lambda$. Let us state definitions and recall basic relationships.
Given probability densities $\{p_i\}_{i=1}^n$ on a space $\mathcal{X}$ with respect a reference measure $\mu$ and $\lambda_i \geq 0$ such that $\sum_{i=1}^n \lambda_i = 1$, define the Bayes risk,
\begin{align} 
    R \coloneqq R_\lambda(p)  1 - \int_\mathcal{X} \max_i \{ \lambda_i p_i (x)\} d\mu(x)
\end{align}
If $\ell(x,y) = 1 -\delta_x(y)$, and we define $T(x) \coloneqq \arg \max_i \lambda_i p_i(x)$ then observe that this definition is consistent with,
the usual definition of the Bayes risk associated to the loss function $\ell$.  Below, we consider $\theta$ to be a random variable on $\{1,2,\dots, n\}$ such that $\mathbb{P}(\theta = i) = \lambda_i$,
and $x$ to be a variable with conditional distribution $\mathbb{P}(X \in A| \theta = i) = \int_A p_i(x) d\mu(x)$.     The following result shows that the Bayes risk gives the probability of the categorization error, under an optimal estimator.
\begin{prop}
    The Bayes risk satisfies
    \[
        R = \min_{\hat{\theta}} \mathbb{E} \ell(\theta, \hat{\theta}(X)) = \mathbb{E} \ell(\theta, T(X))
    \]
    where the minimum is defined over $\hat{\theta}: \mathcal{X} \to \{1,2,\dots, n\}$.
    \end{prop}
    \begin{proof}
    Observe that $R  = 1 - \int_\mathcal{X} \lambda_{T(x)} p_{T(x)}(x) d \mu(x) = \mathbb{E} \ell(\theta, T(X))$.
    Similarly,
        \begin{align*}
            \mathbb{E} \ell(\theta, \hat{\theta}(X))
                % &=
                %     1 - \int_\mathcal{X} \sum_i \lambda_i p_i(x) \mathbbm{1}_{ \{y:i = \hat{\theta}(y)\}}(x) d\mu(x)
                %         \\
                &=
                    1 - \int_\mathcal{X} \lambda_{\hat{\theta}(x)} p_{\hat{\theta}(x)}(x) d\mu(x)
                        \\
                &\geq
                    1 - \int_\mathcal{X} \lambda_{T(x)} p_{T(x)}(x) d\mu(x) = R,
                %         \\
                % &=
                %     \mathbb{E} \ell( i, T(X))
                %         \\
                % &=
                %     R 
        \end{align*}
        which gives our conclusion.
    \end{proof}
    
   The Bayes risk can also be tied directly
    to the total variation, in the following special case.

    \begin{prop}
    When $n =2$ and $\lambda_1 = \lambda_2 = \frac 1 2$, the Bayes risk associated
    to the densities $p_1$ and $p_2$ satisfies
    \begin{align}
        2R = 1 - |p_1 - p_2|_{TV}
    \end{align}
    \end{prop}
    \begin{proof}
    Since $p_T = \frac{|p_1 - p_2| + p_1 + p_2}{2}$, integrating gives $\int_\mathcal{X} p_T(x) d\mu(x) = |p_1-p_2|_{TV} + 1$ from which the equality follows.
    \end{proof}

Information theoretic bounds to control the Bayes and minimax risk have an extensive literature, 
see for example \cite{birge2005new, chen2016bayes, guntuboyina2011lower, xu2016information, YB99}.  Fano's inequality is the
  seminal result in this direction, and we direct the reader to a
 survey of such techniques in statistical estimation, see \cite{scarlett2019introductory}.  What follows can be understood as a sharpening of \cite{guntuboyina2011lower} under the assumption of a $\kappa$-convexity.

    % In this section we assume that $p_i$ are probability measures densities on a Polish space $\mathcal{X}$ with respect to a fixed 
    % reference measure $\mu$, with weights $\lambda_i \geq 0$ such that $\sum_{i=1}^n \lambda_i = 1$ and that $q$ is a probability density with respect to $\mu$ as well.  We recall 
    The function
    $T(x) = \arg \max_i \{ \lambda_i p_i(x)\}$, induces the following convex decompositions of our densities.   The density $q$ can be realized as a convex combination of $q_1 = \frac{\lambda_T q}{1-Q}$ where $ Q = 1 - \int \lambda_T q d\mu$
     and $q_2= \frac{1-\lambda_T) q}{ Q}$,
    \[
        q = (1-Q) q_1 + Q q_2.
    \]
    If we take $p \coloneqq \sum_i \lambda_i p_i$ then $p$ can be decomposed as $\rho_1 = \frac{\lambda_T p_T}{ 1-R}$ and $\rho_2 = \frac{p - \lambda_T p_T}{ R}$ so that
    \[
        p = (1-R ) \rho_1 + R \rho_2.
    \]
    
\begin{thm} \label{thm: the theorem2}
When $f$ is $\kappa$-convex, on $(a,b)$ with $a = \inf_{i,x} \frac{p_i(x)}{q(x)}$ and $b = \sup_{i,x} \frac{p_i(x)}{q(x)}$
\[
    \sum_i \lambda_i D_f( p_i || q) \geq D_f(R||Q) + \frac{\kappa W} 2
\]
where 
\[
   W \coloneqq W(\lambda_i, p_i, q) \coloneqq \frac{(1-R)^2}{1-Q} \chi^2(\rho_1 || q_1) + \frac{R^2}{Q} \chi^2(\rho_2 || q_2) + W_0
\]
for $W_0 \geq 0$.
\end{thm}
$W_0$ can be expressed explicitly as
\[
    W_0 = \int (1- \lambda_{T})Var_{\lambda_i \neq T}\left( \frac{p_i}{q}\right) d \mu =
     \int \sum_{i\neq T} \lambda_i \frac{|p_i - \sum_{j \neq T} \frac{\lambda_j}{1- \lambda_T} p_j|^2 }{q} d\mu,
\]
where for fixed $x$, we consider the variance $Var_{\lambda_i \neq T}\left( \frac{p_i}{q}\right)$ to be the variance of a random variable taking values $p_i(x)/q(x)$ with probability $\lambda_i/(1-\lambda_{T(x)})$ for $i \neq T(x)$.
Note this term is a non-zero term only when $n > 2$.
\begin{proof}
 For a fixed $x$, we apply Lemma \ref{lem: kappa jensens}
\begin{align}
    \sum_i \lambda_i f\left( \frac{p_i}{q} \right)
        &=
            \lambda_T f\left( \frac{p_T}{q} \right) + (1- \lambda_T) \sum_{i \neq T} \frac{\lambda_i}{1 - \lambda_T} f\left( \frac{p_i}{q} \right)
                \\
        &\geq
            \lambda_T  f\left( \frac{p_T}{q} \right) + (1- \lambda_T) \left[ f\left(\frac{p - \lambda_T p_T }{q(1 - \lambda_T)}  \right) + \frac \kappa 2 \Var_{\lambda_{i\neq T}}\left( \frac{p_i}{q}\right) \right]
\end{align}
Integrating,
\begin{align}
    \sum_i \lambda_i D_f (p_i || q) 
        \geq
            \int   \lambda_T  f\left( \frac{p_T}{q} \right) q &+ \int (1- \lambda_T)  f\left(\frac{- \lambda_T p_T + \sum_i \lambda_i p_i }{q(1 - \lambda_T)}  \right)q
        + \frac \kappa 2 W_0,
\end{align}
where
\begin{align}
    W_0 = \int \sum_{i\neq T(x)} \frac{ \lambda_i}{1 - \lambda_T(x)}\frac{|p_i - \sum_{j \neq T} \frac{\lambda_j}{1- \lambda_T} p_j|^2 }{q} d\mu.
\end{align}
Applying the $\kappa$-convexity of $f$,
\begin{align}
    \int   \lambda_T  f\left( \frac{p_T}{q} \right) q
        &= 
            (1-Q) \int q_1 f\left( \frac{p_T}{q} \right) 
                \\
        &\geq (1-Q) \left( f\left( \frac{\int \lambda_T p_T}{1-Q} \right)  + \frac \kappa 2 \Var_{q_1}\left( \frac{p_T}{q} \right) \right)
                \\
        &=
            (1-Q) f((1-R)/(1-Q)) + \frac{Q \kappa}2 W_1,
\end{align}
with 
\begin{align}
    W_1 &\coloneqq  \Var_{q_1}\left( \frac{p_T}{q} \right)
            \\
          &=
          \left(\frac{1-R}{1-Q} \right)^2 \Var_{q_1}\left( \frac{ \lambda_T p_T}{\lambda_T q} \frac{1-Q}{1-R} \right)
            \\
        &=
            \left( \frac{1-R}{1-Q} \right)^2 \Var_{q_1}\left(\frac{\rho_1}{q_1} \right)
                \\
        &=
            \left(\frac{1-R}{1-Q}\right)^2 \chi^2(\rho_1|| q_1)
\end{align}
Similarly,
\begin{align}
    \int (1- \lambda_T)  f\left(\frac{p - \lambda_T p_T  }{q(1 - \lambda_T)}  \right)q
        &=
            Q\int q_2 f\left(\frac{p- \lambda_T p_T  }{q(1 - \lambda_T)}\right)
                \\
        &\geq 
            Q f\left(\int q_2 \frac{p- \lambda_T p_T  }{q(1 - \lambda_T)}\right) + \frac{Q \kappa} 2 W_2
                \\
        &=
            Q f \left( \frac{R}{1-Q}\right) + \frac{Q \kappa} 2 W_2
\end{align}
where
\begin{align}
    W_2
        &\coloneqq \Var_{q_2}\left( \frac{ p - \lambda_T p_T}{q(1-\lambda_T)} \right)
            \\
        &=
            \left(\frac{ R}{ Q} \right)^2 \Var_{q_2} \left( \frac{ p - \lambda_T p_T}{q(1-\lambda_T)} \frac{ Q}{R} \right) 
                \\
       &=
        \left( \frac{R}{Q} \right)^2 \Var_{q_2} \left( \frac{ p - \lambda_T p_T}{q(1-\lambda_T)}  - \frac{ R}{Q} \right)^2
            \\
        &=
            \left( \frac{R}{Q}\right)^2 \int q_2 \left( \frac{\rho_2}{q_2} - 1 \right)^2
                \\
        &=
            \left( \frac{R}{Q}\right)^2 \chi^2(\rho_2|| q_2)
\end{align}
Writing $W= W_0 + W_1 + W_2$ we have our result.
\end{proof}

\begin{cor} \label{cor: Guntuboyina for balanced lambda}
   When $\lambda_i = \frac 1 n$, and $f$ is $\kappa$-convex on $(\inf_{i,x} p_i/q, \sup_{i,x} p_i/q)$
    \begin{align}
    \frac 1 n \sum_i &D_f( p_i ||q) 
    	\\
    		&\geq  D_f(R || (n-1)/n) + \frac{\kappa}{2} \left( n^2(1-R)^2 \chi^2(\rho_1||q) + \left( \frac{nR}{n-1}\right)^2 \chi^2(\rho_2||q) + W_0 \right) 
    \end{align}
    further when $n =2$, 
    \begin{align}
        \frac{D_f(p_1 ||  q) + D_f(p_2 || q)}{2} \geq  D_f &\left(\frac{1-|p_1-p_2|_{TV}}{2}\bigg| \bigg|\frac 1 2 \right) 
            \\
            &+ \frac \kappa 2 
        \left( (1 + |p_1 - p_2|_{TV})^2 \chi^2(\rho_1||q) + (1 - |p_1 - p_2|_{TV})^2 \chi^2(\rho_2 || q) \right) 
    \end{align}
\end{cor}

\begin{proof}
  Note that $q_1 = q_2 = q$, since $\lambda_i = \frac 1 n$ implies $\lambda_T = \frac 1 n$ as well. Also, $Q = 1 - \int \lambda_T q d\mu = \frac{n-1} n$ so that applying
  Theorem \ref{thm: the theorem2} gives,
\begin{align}
   \sum_{i=1}^n D_f( p_i || q) \geq n D_f(R||(n-1)/n)  + \frac{\kappa n W(\lambda_i,p_i,q)}{2}.
\end{align}
The term $W$ can be simplified as well.  In the notation of the proof of Theorem \ref{thm: the theorem2},
\begin{align}
    W_1 &= n^2 (1-R)^2 \chi^2(\rho_1, q)
        \\
    W_2 &= \left( \frac{n R}{n-1} \right)^2 \chi^2(\rho_2 || q)
        \\
    W_0 &= \int \frac{ \frac{1}{n-1} \sum_{i \neq T} (p_i - \frac{1}{n-1}\sum_{j \neq T} p_j )^2}{q} d \mu.
\end{align}
For the special case one needs only to recall $R = \frac{ 1 - |p_1 - p_2|_{TV}}{2}$ while inserting $2$ for $n$.
\end{proof}

%\section{Examples}
%\subsection{Relative Entropy}
\begin{cor} \label{cor: relative entropy example}
When $p_i \leq q /t^*$ for $t^*>0$, and $f(x) = x \log x$
\[
\sum_i \lambda_i D( p_i || q) \geq D(R||Q)  + \frac{t^* W(\lambda_i,p_i,q)} 2
\]
for $D( p_i ||q)$ the relative entropy.  In particular, 
\[
    \sum_i \lambda_i D( p_i || q) \geq  D(p||q) + D(R|| P) + \frac{t^* W(\lambda_i,p_i,p)} 2  
\] 
where $P = 1 - \int \lambda_T p d\mu$ for $p = \sum_i \lambda_i p_i$ and $t^* = \min \lambda_i$.
\end{cor}

\begin{proof}
For the relative entropy, $f(x) = x \log x$ is $\frac 1 M$-convex on $[0,M]$ since $f''(x) = 1/x$.  When $p_i \leq q / t^*$ holds for all $i$ then we can apply Theorem \ref{thm: the theorem2} with $M = \frac 1 {t^*}$.
For the second inequality, recall the
 compensation identity, $\sum_i \lambda_i D(p_i ||q) = \sum_i \lambda_i D(p_i||p) + D(p||q)$, and apply 
 the first inequality to $\sum_i D(p_i ||p)$ for the result. 
%  If we further write $P = 1 = \int \lambda_T q d \mu$, taking $\lambda_i = \frac 1 N$ we have $X$
% \begin{align}
%     \sum_i D(p_i || q) \geq D(R || \frac{N-1}{N}) + \frac{ \chi^2( \rho_1|| p)}{2} + \frac{\chi^2(\rho_2,p)}{2}
% \end{align}
\end{proof}
This gives an upper bound on the Jensen-Shannon divergence, defined as $JSD( \mu|| \nu) = \frac 1 2 D(\mu||\mu/2 + \nu/2) + \frac 1 2 D(\nu || \mu/2 + \nu/2)$.  
Let us also note that through the compensation identity $\sum_i \lambda_i D(p_i||q) = \sum_i \lambda_i D(p_i||p ) + D(p||q)$, $\sum_i \lambda_i D(p_i|| q) \geq \sum_i \lambda_i D(p_i||p)$ where $p = \sum_i \lambda_i p_i$.  In the case that $\lambda_i = \frac 1 N$ 
\begin{align}
	\sum_i \lambda_i &D( p_i || q) 
		\\
		&\geq
			\sum_i \lambda_i D(p_i || p )
				\\
		&\geq
			Q f \left( \frac{1-R}{Q} \right) + (1-Q) f\left( \frac{R}{1-Q} \right) + \frac{t^* W} 2	
\end{align}

\begin{cor} \label{cor: Jensen-Shannon lower bound}
    For two densities $p_1$ and $p_2$, The Jensen-Shannon Divergence satisfies the following,
    \begin{align}
        JSD(p_1||p_2) 
            \geq 
                D &\left(\frac{1 + |p_1 - p_2|_{TV}}{2} \bigg|\bigg| 1/2\right) \\
                    &+ \frac 1 4 \left( (1 + |p_1 - p_2|_{TV})^2 \chi^2(\rho_1||p) + (1 - |p_1 - p_2|_{TV})^2 \chi^2(\rho_2 || p) \right) 
    \end{align}
    with $\rho(i)$ defined above and $p = p_1/2 + p_2/2$.
\end{cor}

\begin{proof}
By Corollary \ref{cor: relative entropy example}, insert $\kappa = 1/2$ into the $n=2$ example of Corollary \ref{cor: Guntuboyina for balanced lambda}.
\end{proof}
Note that $ 2 D((1+V)/2||1/2) = (1+V)\log(1+V) + (1-V) \log (1-V) \geq  V^2 \log e$, we see that a further bound,
\begin{align}   
    JSD(p_1||p_2) 
            \geq 
               \frac{\log e }{2} V^2
                    +  \frac{(1 + V)^2 \chi^2(\rho_1||p) + (1 - V)^2 \chi^2(\rho_2 || p)}{4},
\end{align}
can be obtained for $V = |p_1 - p_2|_{TV}$.

% \subsection{Pearson's chi-squared divergence}
% $f(x) = x^2-1$ is $2$-convex

% \subsection{Squared Hellinger distance} 

% \subsection{Reverse KL divergence}
% $f(x) = - \log x$, on $(0,M]$, $\kappa = \frac 1 {M^2}$

% \subsection{Dual f-divergence}
% If $f$ is $\kappa$-convex on $[a,b]$, then $x \mapsto x f(x^{-1})$ is $\kappa a^3$-convex on $[\frac 1 b, \frac 1 a]$.

% \subsection{Neyman's chi-squared divergence}
% $f(x) = \frac 1 x - 1$ on $[0,M]$ is $2/ M^3$-convex, as can be seen by direct computation, or since Pearson's divergence is $2$ convex on $[1/M, \infty)$.

% \subsection{Another}

% \begin{enumerate}
%     \item 
%     $f(x) = x^l - 1$ for $l > 2$, on $[M, \infty)$, $f''(x) = l(l-1) x^{l-2} \geq \kappa = l(l-1) M^{l-2}$ 
%     \item
%     $f(x) = 1 - x^l$ for $l \in (1,2)$ on $[0,M]$, $\kappa = l (1-l) M^{l-2}$
%     \item
%     $f(x) = x^l -1$ for $l<0$ on $[0,M]$, $\kappa = l(l-1) M^{l-2}$
% \end{enumerate}
 
%  \subsection{Sason's s-divergence}
%  $f_s(x) = (s + t)^2 \log (s +t ) - (s + 1)^2 \log (s +1)$ 
%  for $s \geq e^{-\frac 3 2}$ on $[M, \infty)$ with 
%  $\kappa = 2 \log (s + M) + 3$.  
% \begin{prop}
%     For $f_s(x) = (s+t)^2 \log(s+t) - 
%     (s+1)^2 \log (s+1)$ with $s > e^{-\frac 3 2}$
% \end{prop}

 \subsection{On Topsoe's sharpening of Pinsker's inequality}
 
 For $P_i, Q$ probability measures with densities $p_i$ and $q$ with respect to a common reference measure, $\sum_{i=1}^n t_i = 1$, with $t_i >0$, denote $P = \sum_i t_i P_i$, with density
 $p = \sum_i t_i p_i$, the compensation identity is
 \begin{align} \label{eq: Compensation identity}
     \sum_{i=1}^n t_i D(P_i || Q) = D(P||Q) + \sum_{i=1}^n t_i D(P_i || P).
 \end{align}

 \begin{thm}
     For $P_1$ and $P_2$, denote $M_k = 2^{-k} P_1 + (1- 2^{-k}) P_2$, and define 
     \begin{align*}
         \mathcal{M}_1(k) = \frac{M_k \mathbbm{1}_{\{P_1 > P_2\}} + P_2 \mathbbm{1}_{\{P_1 \leq P_2 \}}}{M_k \{P_1 > P_2\} + P_2 \{P_1 \leq P_2 \}} \hspace{10mm}  \mathcal{M}_2(k) = \frac{M_k \mathbbm{1}_{\{P_1 \leq P_2\}} + P_2 \mathbbm{1}_{\{P_1 >  P_2 \}}}{M_k \{P_1 \leq P_2\} + P_2 \{P_1 > P_2 \}},
     \end{align*}
     then the following sharpening of Pinsker's inequality can be derived,
     \begin{align*}
         D(P_1||P_2) \geq  (2  \log e) |P_1 - P_2 |^2_{TV}  +  \sum_{k=0}^\infty 2^k  \left(\frac{\chi^2(\mathcal{M}_1(k), M_{k+1})}{2} + \frac{\chi^2(\mathcal{M}_2(k), M_{k+1})}{2} \right).
     \end{align*}
 \end{thm}

\begin{proof}
When $n = 2$ and $t_1 = t_2= \frac 1 2$ if we denote $M = \frac{P_1 + P_2}{2}$ then \eqref{eq: Compensation identity} reads as
 \begin{align}
     \frac 1 2 D(P_1 || Q) + \frac 1 2 D(P_2 || Q) = D(M || Q) + JSD(P_1||P_2).
 \end{align}
 Taking $Q = P_2$ we arrive at
 \begin{align}
     D(P_1 || P_2) = 2 D(M || P_2) + 2 JSD(P_1 ||P_2) 
 \end{align}
 Iterating, and writing $M_k = 2^{-k} P_1 + (1- 2^{-k})P_2$, we have
 \begin{align}
     D(P_1 || P_2) = 2^n \left(D(M_n || P_2) + 2 \sum_{k=0}^n JSD(M_n ||P_2) \right)
 \end{align}
 It can be shown (see \cite{topsoe2000some}) that $2^n D(M_n ||P_2) \to 0$ with $n \to \infty$, giving the following series representation,
 \begin{align}
     D(P_1 || P_2) = 2 \sum_{k=0}^\infty 2^k JSD(M_k || P_2).
 \end{align}
Note that the $\rho$-decomposition of $M_k$ is exactly $\rho_i = \mathcal{M}_k(i)$ thus by Corollary \ref{cor: Jensen-Shannon lower bound},
 \begin{align}
     D(P_1 || P_2) 
        &= 
            2 \sum_{k=0}^\infty 2^k JSD(M_k || P_2)
                \\
        &\geq
            \sum_{k=0}^\infty 2^k \left( | M_k - P_2 |^2_{TV} \log e  + \frac{\chi^2(\mathcal{M}_1(k), M_{k+1})}{2} + \frac{\chi^2(\mathcal{M}_2(k), M_{k+1})}{2} \right)
                \\
        &=
        (2  \log e)|P_1 - P_2 |^2_{TV} +  \sum_{k=0}^\infty 2^k  \left(\frac{\chi^2(\mathcal{M}_1(k), M_{k+1})}{2} + \frac{\chi^2(\mathcal{M}_2(k), M_{k+1})}{2} \right).
 \end{align}
 Thus we arrive at the desired sharpening of Pinsker's inequality.
 \end{proof}
 
 Observe that the $k =0$ term in the above series is equivalent to
\begin{align} 
    2^0  \left(\frac{\chi^2(\mathcal{M}_1(0), M_{0+1})}{2} + \frac{\chi^2(\mathcal{M}_2(0), M_{0+1})}{2} \right)
        =
        \frac{\chi^2(\rho_1, p)}{2} + \frac{\chi^2(\rho_2, p)}{2},
\end{align}
where $\rho_i$ is the convex decomposition of $p = \frac{p_1 + p_2}{2}$ in terms of $T(x) = \arg \max \{p_1(x), p_2(x)\}$.

 \pagebreak
 
\appendix
\section{Appendix}

% \begin{defn}
%     Given a convex function $f: [0,\infty) \to \mathbb{R}$ with $f(1) = 0$ and its associated divergence $D_f(\cdot || \cdot )$, define the $r,t$-skew of $D_f$ by
%     \begin{align} \label{eq: skew divergence defintion}
%         S_{f,r,t}(\mu || \nu) \coloneqq D_f( r \mu + (1-r) \nu || t \mu + (1-t) \nu).
%     \end{align}
    
%         It can be shown that for $t\in (0,1)$,  $S_{f,r,t}(\mu || \nu) <\infty.$
% \end{defn}

\begin{thm}  \label{thm: skewing preserves f divergence}
    The class of $f$-divergences is stable under skewing.  That is, if $f$ is convex, satisfying $f(1) =0$, then
    \begin{align}
        \hat{f}(x) \coloneqq (tx + (1-t))f\left(\frac {rx +(1-r)} {tx + (1-t)} \right)
    \end{align}
    is convex with $\hat{f}(1)=0$ as well. %so that the $r,t$ skew of $D_f$ defined in \eqref{eq: skew divergence defintion} is an $f$-divergence as well. 
\end{thm}

\begin{proof}
    If $\mu$ and $\nu$ have respective densities $u$ and $v$ with respect to a reference measure $\gamma$, then $r \mu + (1-r) \nu$ and $t \mu + 1-t \nu$ have densities $r u + (1-r) v$ and $t u + (1-t)v$
    \begin{align}
        S_{f,r,t}( \mu || \nu) 
            &=
                \int f\left(\frac{ r u +(1-r)v}{t u + (1-t) v} \right) (t u + (1-t) v) d\gamma
                    \\
            &=
                \int f \left( \frac{r \frac u v + (1-r) }{t \frac u v + (1-t)} \right) (t \frac u v + (1-t)) v d \gamma
                    \\
            &=
                \int \hat{f}\left( \frac u v \right) v d\gamma.
    \end{align}
    Since $\hat{f}(1) = f(1) = 0$, we need only prove $\hat{f}$ convex.  For this, recall that the conic transform $g$ of a convex function $f$ defined by $g(x,y) = yf(x/y)$ for $y >0$ is convex, since 
    \begin{align}
        \frac{y_1+y_2}{2} f\left(\frac{x_1 + x_2 }{2}/ \frac{y_1 + y_2}{2} \right) 
            &= 
                \frac{y_1+y_2}{2} f\left(\frac{y_1}{y_1+y_2} \frac{x_1}{y_1} + \frac{y_2}{y_1+y_2} \frac{x_2}{y_2} \right) 
                    \\
            &\leq \frac{y_1}{2} f(x_1/y_1) + \frac{y_2}{2} f(x_2/y_2).
    \end{align}
   Our result follows since $\hat{f}$ is the composition of the affine function $A(x) = (rx + (1-r), tx + (1-t))$ with the conic transform of $f$,
    \begin{align}
        \hat{f}(x) = g(A(x)).
    \end{align}
\end{proof}

\bibliographystyle{plain}
\bibliography{bibibi}
\end{document}